\documentclass[10pt,twocolumn,twoside]{IEEEtran}
\usepackage{resizegather}
\usepackage{multirow}
\usepackage{bigstrut}
\usepackage{graphicx}
\usepackage[abs]{overpic}
\usepackage{algorithmic}
\usepackage{epstopdf}
\usepackage{setspace}

\usepackage{amsmath}
\usepackage{amsfonts}

\usepackage{mathtools}
\usepackage{tikz}
\usepackage{color}
\usepackage{amsthm} 
\newtheorem{theorem}{\bf{Theorem}}[section]

\newtheorem{cor}[theorem]{Corollary}
\newtheorem{lem}[theorem]{Lemma}

\theoremstyle{plain}

\newenvironment{definition}[1][Definition]{\begin{trivlist}
\item[\hskip \labelsep {\bfseries #1}]}{\end{trivlist}}

\newcounter{assump} 
\newcounter{rem} 
\newenvironment{remark}[1][Remark \arabic{rem}]{\refstepcounter{rem} \begin{trivlist} 
\item[\hskip \labelsep {\bfseries #1}]}{\end{trivlist}}

\newcounter{propno} 
\mathtoolsset{showonlyrefs}

\newcounter{algno} 
\title{\LARGE \bf
Structural Robustness to Noise in Consensus Networks: Impact of Degrees and Distances, Fundamental Limits, and Extremal Graphs
}
\author{ Yasin Yaz{\i}c{\i}o\u{g}lu, \IEEEmembership{Member, IEEE},  Waseem Abbas, \IEEEmembership{Member, IEEE}, and Mudassir Shabbir
\thanks{Yasin~Yaz{\i}c{\i}o\u{g}lu is with the Department of Electrical and Computer Engineering at the University of Minnesota, Minneapolis, MN, USA (e-mail: ayasin@umn.edu).} 
\thanks{Waseem~Abbas is with the Department of Electrical Engineering and Computer Science, Vanderbilt University, Nashville, TN, USA (e-mail: waseem.abbas@vanderbilt.edu).}
\thanks{Mudassir~Shabbir is with the Computer Science Department at the Information Technology University, Lahore, Punjab, Pakistan (e-mail: mudassir@rutgers.edu).}

\thanks{Some preliminary results of this paper were presented in the $58^{\text{th}}$ IEEE Conference on Decision and Control (see  \cite{Yasin19CDC}).}

}

\IEEEoverridecommandlockouts
\begin{document}

\maketitle

\begin{abstract}

We investigate how the graph topology influences the robustness to noise in undirected linear consensus networks. We measure the structural robustness by using the smallest possible value of steady state population variance of states under the noisy consensus dynamics with edge weights from the unit interval. We derive tight upper and lower bounds on the structural robustness of networks based on the average distance between nodes and the average node degree. Using the proposed bounds, we characterize the networks with different types of robustness scaling under increasing size. Furthermore, we present a fundamental trade-off between the structural robustness and the average degree of networks.  While this trade-off implies that a desired level of structural robustness can only be achieved by graphs with a sufficiently large average degree, we also show that there exist dense graphs with poor structural robustness.  We then show that, random $k$-regular graphs (the degree of each node is $k$) with $n$ nodes typically have near-optimal structural robustness among all the graphs with size $n$ and average degree $k$  for sufficiently large $n$ and $k$. We also show that when $k$ increases properly with $n$, random $k$-regular graphs maintain a structural robustness within a constant factor of  the complete graph's while also having the minimum average degree required for such robustness.


\end{abstract}

%%%%%%%%%%%%%%%%%%%%%%%%%%%%%%%%%%%%%%%%%%%%%%%%%%%%%%%%%%%%%%%%%%%%%%%%%%%%%%%%

%%%%%%%%%%%%%%%%%%%%%%%%%%%%%%%%%%%%%%%%%%%%%%%%%%%%%%%%%%%%%%%%%%%%%%%%%%%%%%%%
\section{Introduction}
Consensus networks, where the state of each node approaches a weighted average of the states of adjacent nodes, are used to model the diffusive couplings numerous natural and engineered systems. These systems typically operate in the face of various disturbances such as measurement/process noise, communication delays, component failures, misbehaving nodes, or malicious attacks (e.g., \cite{Young10, Bamieh12,Leblanc13,Shi13,yazicioglu2017resilient}).
Accordingly, a central question regarding such networks is how they behave in the face of disturbances.

%Graph measures such as connectivity (e.g., \cite{Dekker04, Jamakovic07}), expansion ratios (e.g., \cite{Malliaros12,Yasin15TNSE}), and Kirchoff index (e.g., \cite{Young10,Bamieh12,Abbas12}) have been used in the literature to quantify the robustness to different types of disturbances. 

This paper is focused on the robustness of undirected consensus networks to noisy interactions. In such networks, each edge is endowed with some positive weight denoting the coupling strength between the corresponding nodes. We consider a setting with additive process noise, where the state of each node is attracted towards the weighted average of the states of its neighbors plus some independent and identically distributed (i.i.d.) white Gaussian noise with zero mean and unit covariance. We use the expected steady state population variance of states, which is a variant of the $\mathcal{H}_2$-norm of the system with the output defined as the deviation of nodes from global average, as the measure of vulnerability to noise. Similar dynamics were considered in \cite{Young10,Bamieh12} and it was shown that for any network with a given allocation of edge weights, the expected steady state variance can be expressed in terms of the weighted Laplacian eigenvalues. Some tight bounds on this robustness measure were presented in \cite{Siami16,Siami17}. 

%In \cite{Patterson10} and \cite{Lin14}, the authors investigated the use of leader-follower control for improving the robustness of noisy consensus networks and presented algorithms for optimal leader selection. 

In this paper, we introduce the notion of structural robustness to noise,  which extends the related measures in the literature (e.g., \cite{Young10,Bamieh12}) and assess each network based on the smallest value of expected steady state variance that can be attained under the noisy consensus dynamics with edge weights from the unit interval.  We show that two simple graph measures, namely the average distance between nodes and the average node degree, define tight bounds on the proposed measure of structural robustness. We then use these bounds to obtain some fundamental limits and trade-offs regarding structural robustness and to characterize graphs with extremal robustness. The main contributions of this paper are as follows:

% \color{blue}
%  Theorem \ref{bounds} is closely related to Theorems 6 and 8 in \cite{Siami16}, which provide two similar bounds on a robustness measure that is equal to $n\mathcal{H}^*(\mathcal{G})$: a lower bound based on the degree sequence (all the node degrees) and an upper bound based on the diameter and the number of edges. While the degree sequence may yield a better (closer to the actual value) bound than the average degree, the average degree can be computed with significantly less information about the graph, namely the number of nodes and the number of edges as per \eqref{degav2}. The corresponding lower bounds are close when all the nodes have similar degrees and they are equal for regular graphs.  Furthermore, while the upper bound in Theorem \ref{bounds} can be computed based on only the pairwise distances between nodes, computing the upper bound in \cite[Thm. 6]{Siami16} requires also the number of edges. It can be shown that these two upper bounds may outperform each other on different graphs. For example, while using the diameter and the number of edges yields a better bound for the complete graph, the average distance gives a better bound for any path graph with $n \geq 4$.  
% \color{black}

\begin{itemize}
%\item  We show that this measure of structural robustness can be expressed in terms of the eigen values of the unweighted Laplacian.

\item  We show that the average distance between nodes and the average node degree define tight upper and lower bounds on the proposed measure of structural robustness to noise. Using these bounds, we also provide a characterization of networks with extremal scaling of structural robustness, i.e., graph families such that the structural robustness gets arbitrarily worse (e.g., path graph) or arbitrarily better (e.g., complete graph) as the network size increases.

\item  We show that there is a fundamental trade-off between the structural robustness and the  edge-sparsity of networks. We express this trade-off in terms of tight bounds on the ratio of structural robustness of any given graph to the structural robustness of the complete graph (best) and the star graph (best among the connected graphs with minimum average degree).  While these bounds imply that a desired level of structural robustness can only be achieved by graphs with a sufficiently large average degree, we also show that there exist graphs whose robustness becomes arbitrarily worse with increasing size despite having an arbitrarily large average degree.  

\item  We show that, for sufficiently large $n$ and $k$, random $k$-regular graphs with $n$ nodes typically have near-optimal structural robustness among the graphs with size $n$ and average degree $k$. Moreover, when $k$ increases properly with size, random $k$-regular graphs maintain a structural robustness within a constant factor of the complete graph's while also having the minimum average degree required for such robustness.

\end{itemize}

%showing that two simple aggregate measures, i.e., the average distance between the nodes and the average node degree, define tight upper and lower bounds on the proposed notion of structural robustness.  In particular, we show that as the size increases the structural robustness can get arbitrarily worse only if the average degree between the nodes grows unbounded, and it can get arbitrarily better only if the average degree grows unbounded. We also show that random regular graphs asymptotically (as size goes to infinity) have structural robustness within a constant factor of the best possible achievable among the graphs with same size and sparsity. Furthermore, this constant factor gets arbitrarily close to one as the average degree increases. 

The organization of this paper is as follows: Section \ref{prelim} provides some graph theory preliminaries. Section \ref{main} presents our main results. Section \ref{sims} provides the numerical simulations. Finally, Section \ref{conclusion} concludes the paper.

\section{ Preliminaries}
\label{prelim}

\subsection{Notation}
We use $\mathbb{R}$ and $\mathbb{R}_{+}$ to denote the set of real numbers and positive real numbers, respectively. For any finite set $A$ with cardinality $|A|$, we use $\mathbb{R}^{|A|}$ (or $\mathbb{R}_{+}^{|A|}$) to denote 
the space of real-valued (or positive-real-valued) ${|A|-\mbox{dimensional}}$ vectors.  For any pair of vectors $x,y \in \mathbb{R}^{|A|}$, we use $x \leq y$ (or $x<y$) to denote the element-wise inequalities, i.e., $x_i \leq y_i$ (or $x_i < y_i$) for all $i=1,2, \hdots, |A|$. 
 The all-ones and all-zeros vectors, their sizes being clear from the context, will be denoted by $\bold{1} \in \mathbb{R}^n$ and $\bold{0} \in \mathbb{R}^n$.  
We use $\mathcal{O}(\cdot)$ to denote the Big O notation.
\subsection{Graph Theory Basics}
A graph $\mathcal{G}=(V,E)$ consists of a node set ${V=\{1,2,\hdots,n\}}$ and an edge set $E \subseteq V \times V$. For an undirected graph, each edge is represented as an unordered pair of nodes. For each $i \in V$, let $\mathcal{N}_i$ denote the \emph{neighborhood} of $i$, i.e., $\mathcal{N}_i = \{ j \in V \mid (i,j) \in E \}$. A \emph{path} between a pair of nodes $i,j \in V$ is a sequence of distinct nodes $\{i, \hdots, j\}$ such that each pair of consecutive nodes are linked by an edge. For any node $i$, the number of nodes in its neighborhood, $|\mathcal{N}_i|$, is called its degree, $d_i$. Accordingly, the average node degree is 
 \begin{equation}
\label{degav2}
\tilde{d}(\mathcal{G})=\frac{1}{n}\sum_{i=1}^nd_i=\frac{2|E|}{n}.
\end{equation}

The \emph{distance} between any two nodes $i$ and $j$, which is denoted by $\delta_{ij}$, is equal to the number of edges on the shortest path between those nodes. The maximum distance between any two nodes, $\max_{i,j \in V} \delta_{ij}$ is known as the \emph{diameter} of the graph, and the \emph{average distance} between the nodes is given as
\begin{equation}
\label{distav}
\tilde{\delta}(\mathcal{G})=\frac{2}{n^2-n}\sum_{1\leq i<j\leq n}\delta_{ij}.
\end{equation}
A graph is \emph{connected} if there exists a path between every pair of nodes. A connected undirected graph with $n$ nodes is called a \emph{tree} if  there is exactly one path between each pair of nodes. Any connected graph with $n$ nodes and $n-1$ edges is a tree.  A graph is called a $k$-regular graph if the number of edges incident to each node (the degree) is equal to $k$. A random $k$-regular graph, $\mathcal{G}_{n,k}$, is a graph that is selected uniformly at random from the set of all $k$-regular graphs with $n$ nodes.

For weighted graphs, we use $w \in \mathbb{R}_+^{|E|}$ to denote the vector of edge weights and $w_{ij}\in \mathbb{R}_+$ to denote the weight of the edge $(i,j)\in E$. The (weighted) \emph{graph Laplacian} of a weighted graph is defined as
\begin{equation}
\label{Deg}
 [L_w]_{ij}=\left\{\begin{array}{ll}\sum_{k\in \mathcal{N}_{i}}w_{ik}&\mbox{ if } 
i=j\\-w_{ij}&\mbox{ if } j\in \mathcal{N}_i \\ 0&\mbox{ otherwise }\end{array}\right.
\end{equation}

In the remainder of the paper, we will use $L$ to denote the unweighted Laplacian, i.e., the special case when $w=\bold{1}$.

\subsection{Consensus Networks}
Consensus networks can be represented as a graph, where the nodes correspond to the agents, and the weighted edges exist between the agents that are coupled through local interactions. For such a network $\mathcal{G}=(V,E)$, let the dynamics of each agent $i\in V$ be
\begin{equation}
\label{consensus}
\dot{x}_i(t)=\sum_{j\in \mathcal{N}_i}w_{ij}(x_j(t)-x_i(t))+ \xi_i(t),
\end{equation}
where $x_i(t) \in \mathbb{R}$ denotes the state of $i$, each ${w_{ij} \in \mathbb{R}_+}$ is a constant weight representing the strength of the coupling between $i$ and $j$, and $\xi(t) \in \mathbb{R}^n$ is i.i.d. white Gaussian noise with zero mean and unit covariance, which is one of the standard noise models for agents that are independently affected by disturbances of same
intensity due to various effects such as communication errors, noisy measurements, or quantization errors (e.g.,\cite{Young10,Bamieh12,Siami16}).  Accordingly, the overall dynamics of the agents can be expressed as 
\begin{equation}
\label{consensus2}
\dot{x}(t)=-L_wx(t)+ \xi(t),
\end{equation}
where $L_w$ denotes the weighted Laplacian. In a noise-free setting ($\xi(t)= \bold{0}$ for all $t\geq 0$), the dynamics in \eqref{consensus2} are known to result in a global consensus,  ${\lim_{t\to \infty} x(t) \in span \{\bold{1}\}}$, for any $x(0) \in \mathbb{R}^n$ if and only if the graph is connected \cite{Jadbabaie03,Ren05}. In the noisy case, a perfect consensus can not be achieved. Instead, some finite steady state variance of $x(t)$ is observed on connected graphs \cite{Young10, Bamieh12}. Accordingly, the robustness of the network can be quantified through the expected population variance in steady state, i.e.,
\begin{equation}
\label{heq}
\mathcal{H}(\mathcal{G},w) \coloneqq \lim_{t \to \infty} \frac{1}{n}\sum\limits_{i =1}^{n} \mathrm{E}[{(x_i(t)-\tilde{x}(t))^2}],
\end{equation}
where $\tilde{x}(t) \in \mathbb{R}$ denotes the average of $x_1(t), x_2(t), \hdots, x_n(t)$.

 It can be shown that (e.g., see \cite{Young10,Bamieh12}) $\mathcal{H}(\mathcal{G},w)$ is equal to $1/n$ times the square of the $\mathcal{H}_2$-norm of the system in \eqref{consensus2} from the input $\xi(t)$ to the output $y(t) \in \mathbb{R}^n$ defined as ${y_i(t)= x_i(t)-\tilde{x}(t)}$, and it satisfies  
\begin{equation}
\label{heq2}
\mathcal{H}(\mathcal{G},w) = \frac{1}{2n} \sum_{i=2}^n \frac{1}{\lambda_i (L_w)},
\end{equation}
where   and $0 < \lambda_2(L_w) \leq \hdots \leq \lambda_n(L_w) $ denote the eigenvalues of the weighted Laplacian $L_w$.

In this paper, we investigate how much the structure of the underlying graph (the edge set $E$) causes vulnerability to noise in consensus networks. We measure the structural vulnerability of any given network to noise based on the smallest possible value of $\mathcal{H}(\mathcal{G},w)$, given that the edge weights should belong to the feasible set ${ \mathcal{W}= \{w  \mid  \bold{0} < w \leq \bold{1}  \}}$. Since multiplying all the weights by some $\alpha \in \mathbb{R}_+$ results in ${L_{\alpha w}=\alpha L_w}$ and ${\mathcal{H}(\mathcal{G},\alpha w)= \mathcal{H}(\mathcal{G},w)/\alpha}$ due to \eqref{heq2}, it is possible to make $\mathcal{H}(\mathcal{G},w)$ arbitrarily small for any network by just scaling up all the weights. By considering only weights in $(0,1]$, we remove this possibility and focus on the impact of network structure. 

\begin{definition}
(Structural Vulnerability and Robustness) The structural vulnerability of an undirected consensus network $\mathcal{G}=(V,E)$ to noise is the smallest possible value of $\mathcal{H}(\mathcal{G},w)$ that is achievable under weights from the unit interval, i.e.,
\begin{equation}
\label{strob}
\mathcal{H}^*(\mathcal{G}) \coloneqq \min_{\bold{0} < w \leq \bold{1}} \mathcal{H}(\mathcal{G},w).
\end{equation}
The structural robustness to noise is quantified using the reciprocal of structural vulnerability, $1/\mathcal{H}^*(\mathcal{G})$.
\end{definition}

 \begin{remark}
 For brevity, we will say ``structural robustness (or vulnerability)" without explicitly saying ``to noise". The term ``structural robustness" is also used in the literature for referring to the robustness of connectivity to node/edge failures (e.g., \cite{Wu11,Abbas12}). While the two notions of robustness have connections, the distinction should be clear from the context.
   \end{remark}

\section{Main Results}
\label{main}
In this section, we provide the main results of this paper. We start our derivations by providing $\mathcal{H}^*(\mathcal{G})$ as a function of the (unweighted) Laplacian eigenvalues.

\begin{lem} \label{lspect}For any connected undirected graph $\mathcal{G}$,   \begin{equation}
\label{hseig}
\mathcal{H}^*(\mathcal{G})= \frac{1}{2n} \sum_{i=2}^n \frac{1}{\lambda_i (L)},
\end{equation} where $L$ denotes the unweighted Laplacian of $\mathcal{G}$.\end{lem}
\begin{proof}
For any connected undirected $\mathcal{G}$, any weighted Laplacian is a positive semidefinite matrix \cite{Merris94}. Increasing any of its weights or adding new edges leads to a new Laplacian that is equal to the initial Laplacian plus another matrix that is also a weighted Laplacian (a graph with just the added/strengthened edges).   All the Laplacian eigenvalues monotonically (not necessarily strictly) increase under such an addition of a positive semidefinite matrix due to the Weyl's inequality (e.g., see \cite{Horn90}). Hence, $\mathcal{H}(\mathcal{G},w)$ is minimized for $w=\bold{1}$ within the feasible set of \eqref{strob}. Accordingly, using \eqref{heq2}, we obtain \eqref{hseig}. 
\end{proof}

In light of Lemma \ref{lspect}, $\mathcal{H}^*(\mathcal{G})$ of any connected network can be computed through the eigenvalues of the unweighted Laplacian.
Furthermore, using this result, $\mathcal{H}^*(\mathcal{G})$ can also be expressed in terms of a graph measure known as the Kirchhoff index (total effective resistance) \cite{Klein93}. For any connected undirected graph with $n$ nodes, $\mathcal{G}$, the Kirchhoff index satisfies
\begin{equation}
\label{eqn:K_f}
K_f(\mathcal{G}) = n\sum\limits_{i=2}^n\frac{1}{\lambda_i(L)},
\end{equation} 
where $L$ is the Laplacian of $\mathcal{G}$. Accordingly, due to \eqref{hseig},
\begin{equation}
\label{srobKf}
\mathcal{H}^*(\mathcal{G})= \frac{K_f(\mathcal{G})}{2n^2}.
\end{equation} 

The connection in  \eqref{srobKf} is particularly useful as it links the structural robustness to the rich literature in graph theory on Kirchhoff index. For instance, closed form expressions in terms of size are known for some graph families (e.g., see \cite{Palacios01,Lukovits99, Ellens11}). Using those results on Kirchhoff index we immediately obtain that the path ($\mathcal{P}_n$), cycle ($\mathcal{C}_n$), star ($\mathcal{S}_n$), and complete ($\mathcal{K}_n$) graphs of size $n$ have
\begin{align}
   \label{robPCSK}
\mathcal{H}^*(\mathcal{P}_n) =\frac{n^2-1}{12n} &\;,\;
\mathcal{H}^*(\mathcal{C}_n) =\frac{n^2-1}{24n}, \\
\mathcal{H}^*(\mathcal{S}_n) = \frac{(n-1)^2}{2n^2} &\;,\;
\mathcal{H}^*(\mathcal{K}_n)=\frac{(n-1)}{2n^2}. 
\end{align}

Furthermore, among all the connected undirected graphs with $n$ nodes, the Kirchoff index is minimized in the complete graph $\mathcal{K}_n$ and maximized in the path graph $\mathcal{P}_n$ (e.g., see \cite{Ellens11}). As such, in light of \eqref{srobKf}, $\mathcal{K}_n$ and $\mathcal{P}_n$ are also the minimizer and maximizer of $\mathcal{H}^*(\mathcal{G})$, respectively. 
%There are also some results in the literature regarding the graphs with minimum Kirchoff index when there are additional constraints on topology. For example, among the graphs with a given size and diameter, the minimum Kirchoff index occurs in clique chains \cite{Ellens11}. On the other hand, among the graphs with a given size and edge cut (graphs that can be disconnected by removing a given number of edges), clique-stars have the minimum Kirchoff index \cite{Deng10}.
\subsection{Impact of Average Degree and Average Distance}

The structural vulnerability of any given network can be computed by using the Laplacian eigenvalues as in \eqref{hseig}. However, it is not easy to use \eqref{hseig} or \eqref{srobKf} for certain analysis and design applications in a systematic and efficient way. For instance, finding an optimal way to add a given number of edges to an arbitrary network to reduce the $\mathcal{H}^*(\mathcal{G})$ would require searching among all possibilities (e.g., see \cite{Ellens11}). Furthermore, while it is possible to see how $\mathcal{H}^*(\mathcal{G})$ scales with size for the special graph families with closed form expressions as in \eqref{robPCSK}, it is hard to analyze the asymptotic robustness of generic networks. One way to overcome these type of difficulties is focusing on some bounds on $\mathcal{H}^*(\mathcal{G})$ rather than its exact value.

Many bounds on the Kirchhoff index have been proposed in the literature by using graph measures such as chromatic number, independence number, edge/node connectivity, diameter, or degree sequence (e.g., see \cite{Zhou08,Milovanovic17}). These bounds typically require significant amount of global information and/or computation, which limits their applicability in large networks. Motivated by such limitations, we present a fundamental relationship between the $\mathcal{H}^*(\mathcal{G})$ and two aggregate measures, namely the average node degree and the average distance between nodes, which can be computed/estimated efficiently based on limited information (e.g. \cite{Goldreich08}). 

%For example, while an improved lower bound can be obtained when every node degree is known (e.g., \cite{Siami16}), the average degree can be computed by using only the number of nodes and the number of edges as per \eqref{degav2}.  
 
\begin{theorem}
\label{bounds}
For any connected undirected graph ${\mathcal{G}=(V,E)}$ with $n\geq 2$ nodes, \begin{equation}
\frac{(n-1)^2}{2\tilde{d}(\mathcal{G})n^2} \leq \mathcal{H}^*(\mathcal{G}) \leq \frac{\tilde{\delta}(\mathcal{G})(n-1) }{4n},
\label{boundseq}
\end{equation}
where $\tilde{d}(\mathcal{G})$ is the average node degree, $\tilde{\delta}(\mathcal{G})$ is the average distance between the nodes. Moreover, the lower bound holds with equality if and only if $\mathcal{G}$ is a complete graph, and the upper bound holds with equality if and only if $\mathcal{G}$ is a tree. 
\end{theorem}
\begin{proof}
(Lower bound:) Since the harmonic mean is always less than or equal to the arithmetic mean, we have
\begin{equation}
\label{lb1}
\frac{n-1}{\sum_{i=2}^n \lambda_i (L)} \leq \frac{1}{n-1}\sum_{i=2}^n \frac{1}{\lambda_i (L)},
\end{equation} 
where the left side is the harmonic mean and the right side is the arithmetic mean of $1/\lambda_2(L), 1/\lambda_3(L), \hdots, 1/\lambda_n(L)$. Furthermore since $L$ is a symmetric matrix, the sum of its eigenvalues equals its trace, which is equal to the sum of node degrees $n\tilde{d}(\mathcal{G})$. Hence, \eqref{lb1} implies
\begin{equation}
\label{lb2}
\frac{(n-1)^2}{n\tilde{d}(\mathcal{G})} \leq \sum_{i=2}^n \frac{1}{\lambda_i (L)}.
\end{equation} 
Due to \eqref{heq2} and \eqref{lb2},
\begin{equation}
\label{lb3}
\mathcal{H}^*(\mathcal{G}) =\frac{1}{2n} \sum_{i=2}^n \frac{1}{\lambda_i (L)} \geq \frac{(n-1)^2}{2\tilde{d}(\mathcal{G})n^2} .
\end{equation} 
Alternatively, \eqref{lb3} can also be obtained by using \cite[Theorem 6]{Siami16} and Lemma \ref{lspect}, which implies that the performance measure in \cite[Theorem 6]{Siami16} equals $n\mathcal{H}^*(\mathcal{G})$.  

Note that the harmonic mean equals the arithmetic mean if and only if all the numbers are equal. Hence, \eqref{lb1} holds with equality if and only if $\lambda_2(L)=\lambda_3(L)= \hdots = \lambda_n(L)$. Furthermore, all the positive Laplacian eigenvalues of a connected graph are equal if and only if the graph is a complete graph (e.g., see \cite{Merris94}). Hence, \eqref{lb3} holds with equality if and only if $\mathcal{G}$ is a complete graph.

(Upper bound:)  The Kirchoff index is defined as the sum of pairwise effective resistances between nodes \cite{Klein93,Ellens11}, i.e.,
\begin{equation}
\label{ub2a}
K_f(\mathcal{G}) = \sum_{1\leq i < j \leq n}r_{ij},
\end{equation}
where $r_{ij}$ is equal to the effective resistance between the nodes $i$ and $j$ on an electrical network that is obtained by assigning a unit resistor to each edge of $\mathcal{G}$. For any two nodes $i$ and $j$, $r_{ij}=\delta_{ij}$ if there is a unique path between $i$ and $j$, and $r_{ij}<\delta_{ij}$ otherwise  (e.g., see \cite[Theorem 2.4]{Ellens11}). Accordingly, the Kirchoff index satisfies 
\begin{equation}
\label{ub2}
K_f(\mathcal{G}) \leq \sum_{1\leq i < j \leq n}\delta_{ij},
\end{equation}
 and \eqref{ub2} holds with equality if and only if there is a unique path between any two nodes, i.e. $\mathcal{G}$ is a tree. Since the sum of distances between the nodes satisfy
\begin{equation}
\label{ub3}
 \sum_{1\leq i < j \leq n}\delta_{ij}= \frac{n(n-1)\tilde{\delta}(\mathcal{G})}{2},\end{equation}
\eqref{srobKf} and \eqref{ub2} together imply
\begin{equation}
\label{ub4}
\mathcal{H}^*(\mathcal{G}) \leq \frac{\tilde{\delta}(\mathcal{G})(n-1)}{4n}.\end{equation}
Furthermore, since \eqref{ub2} holds with equality if and only if $\mathcal{G}$ is a tree, the same is true for the inequality in \eqref{ub4}. Alternatively, the upper bound can also be proved by using \cite[Theorem 2]{Sivasubramanian09} and Lemma \ref{lspect}. 
\end{proof}

 Theorem \ref{bounds} is closely related to \cite[Thms. 6 and 8]{Siami16}, which can be combined with Lemma \ref{lspect} to obtain two other bounds on $\mathcal{H}^*(\mathcal{G})$: a lower bound based on the degree sequence (all the node degrees) and an upper bound based on the diameter. While using the degree sequence may yield a better lower bound (closer to actual value), the average degree can be computed with significantly less information, namely the number of nodes and the number of edges as in \eqref{degav2}. The upper bound in Theorem \ref{bounds} can be computed by using only the pairwise distances between nodes, whereas computing the upper bound in \cite[Thm. 6]{Siami16} also requires the number of edges. These two upper bounds may outperform each other on different graphs. 
For example, while \cite[Thm. 6]{Siami16} yields a better upper bound for the complete graph, Theorem \ref{bounds} gives a better upper bound for the path graph.

\subsection{Graphs with Extremal Robustness Scaling}
One of the important considerations when designing large scale networks is how the robustness of the system would scale with its size. As indicated by \eqref{robPCSK}, different network topologies may exhibit different robustness scaling properties. 
For instance, while the  structural vulnerability  of complete graph, $\mathcal{H}^*(\mathcal{K}_n)$, tends to zero as the network size increases (see \eqref{robPCSK}), the  structural vulnerability  of path graph, $\mathcal{H}^*(\mathcal{P}_n)$, tends to infinity as the network size increases (see \eqref{robPCSK}). Apart from these two extremal cases of robustness scaling, there are also networks (e.g., star graph) such that $\mathcal{H}^*(\mathcal{G}_n)$ converges to some non-zero value as the network size increases. One question of interest is then which topological properties determine how the structural robustness behaves as the size goes to infinity. Our next result provides a graph topological characterization of networks with extremal robustness scaling. 
\begin{cor}
 \label{findeg} Let $\{\mathcal{G}_n\}_{n\in\mathbb{N}}$ denote an infinite sequence of connected undirected graphs with $n$ nodes. The structural vulnerability of $\mathcal{G}_n$ tends to zero as $n$ goes to infinity only if the average node degree grows unbounded, i.e.,  
\begin{equation}
\label{findeg1}
\lim_{n \to \infty}\mathcal{H}^*(\mathcal{G}_n) = 0 \Rightarrow \lim_{n \to \infty} \tilde{d}(\mathcal{G}_n) = \infty.
\end{equation}
Furthermore, the structural vulnerability grows unbounded only if the average distance also grows unbounded, i.e.,  
\begin{equation}
\label{findeg2}
\lim_{n \to \infty}\mathcal{H}^*(\mathcal{G}_n) = \infty \Rightarrow \lim_{n \to \infty} \tilde{\delta}(\mathcal{G}_n) = \infty.
\end{equation}
\end{cor}
\begin{proof}
($\mathcal{H}^*(\mathcal{G}_n) \to 0 $): Note that the lower bound in \eqref{boundseq} is non-negative for any connected undirected $\mathcal{G}$ with $n\geq 2$ nodes. Hence, due to the squeeze theorem, if $\mathcal{H}^*(\mathcal{G}_n)$ tends to zero then the lower bound must also tend to zero, i.e., 
\begin{equation}
\label{findeg2b}
\lim_{n \to \infty}\mathcal{H}^*(\mathcal{G}_n) = 0 \Rightarrow \lim_{n \to \infty} \frac{(n-1)^2}{2\tilde{d}(\mathcal{G}_n)n^2}= 0.
\end{equation}
Since the average node degree $\tilde{d}(\mathcal{G}_n)$ is positive, \eqref{findeg2b} implies
\begin{equation}
\label{findeg3}
\lim_{n \to \infty}\tilde{d}(\mathcal{G}_n)= \infty.
\end{equation}

($\mathcal{H}^*(\mathcal{G}_n) \to \infty $):   If $\mathcal{H}^*(\mathcal{G}_n)$ diverges as $n$ goes to infinity, the upper bound in \eqref{boundseq} must also diverge, which is only possible if the average distance between nodes, $\tilde{\delta}(\mathcal{G}_n)$, diverges. 
\end{proof}

\subsection{Structural Robustness vs. Sparsity}

We first highlight a fundamental trade-off between structural robustness and sparsity. We use the average node degree as the measure of sparsity (lower average degree implies higher sparsity). We express this trade-off in terms of tight bounds, i.e., bounds that are satisfied with equality for some connected $\mathcal{G}_n$, on the ratio of structural robustness of any given graph to the structural robustness of the complete graph, which has the best robustness among all connected graphs, and the star graph, which has the best structural robustness achievable with the minimum number of edges a connected graph can have.

\begin{theorem}
\label{gn-hstar-th}For any connected undirected graph $\mathcal{G}_n$,
\begin{equation}
\label{gn-hstar}
\frac{\mathcal{H}^*(\mathcal{S}_n)}{\mathcal{H}^*(\mathcal{G}_n)}\leq \tilde{d}(\mathcal{G}_n),
\end{equation}
\begin{equation}
\label{gn-hcomp}
\frac{\mathcal{H}^*(\mathcal{G}_n)}{\mathcal{H}^*(\mathcal{K}_n)}\geq \frac{n-1}{\tilde{d}(\mathcal{G}_n)},
\end{equation}
where  $\mathcal{S}_n$ and $\mathcal{K}_n$ denote the star and complete graphs with $n$ nodes. Furthermore, these inequalities are tight in the sense that they hold with equality for some $\mathcal{G}_n$.
\end{theorem}
\begin{proof}
Both \eqref{gn-hstar} and \eqref{gn-hcomp} follow from \eqref{robPCSK} and the lower bound in \eqref{boundseq}. The tightness of the bounds can be proved by showing that they are satisfied with equality for some $\mathcal{G}_n$. For instance, the bounds are satisfied with equality for the complete graph, $\mathcal{G}_n=\mathcal{K}_n$, due to \eqref{robPCSK} and the fact that $\tilde{d}(\mathcal{K}_n)=n-1$.
\end{proof}

Since $\mathcal{S}_n$ has the best structural robustness achievable with the minimum number of edges a connected graph can have, \eqref{gn-hstar} highlights the price of structural robustness in terms of sparsity. Any graph with significantly better structural robustness than the star graph of same size should have a proportionally high average degree. Similarly, \eqref{gn-hcomp} indicates how sparse a graph can be while having a certain level of structural robustness relative to the complete graph.

Theorem \ref{gn-hstar-th} can be used for the design of sparse yet robust networks. For example, consider a network design problem, where the goal is to build a network with the minimum number of edges that has a bounded robustness-suboptimality with respect to the complete graph, i.e., ${\mathcal{H}^*(\mathcal{G}_n)/\mathcal{H}^*(\mathcal{K}_n) \leq \alpha}$ for some desired $\alpha \geq 1$. For instance, in a wireless sensor network, this design problem can be motivated by the goal of achieving robust distributed estimation with minimum communication due to energy and  bandwidth considerations. In light of \eqref{gn-hcomp}, such a network must have an average degree of at least $(n-1)/\alpha$, i.e.,
\begin{equation}
\label{gn-hcomp-cor}
\frac{\mathcal{H}^*(\mathcal{G}_n)}{\mathcal{H}^*(\mathcal{K}_n)}\leq \alpha \Rightarrow \tilde{d}(\mathcal{G}_n) \geq \frac{n-1}{\alpha}.
\end{equation}
Accordingly, the design space can be narrowed down to the set of graphs with sufficiently many edges as per \eqref{gn-hcomp-cor}. Note that \eqref{gn-hcomp-cor} defines a necessary condition on sparsity and not every graph with that many edges have the desired robustness property. In fact, our next result shows that there even exist graphs whose structural vulnerability grows unbounded despite having such an average degree. This result also complements Corollary \ref{findeg} by showing that an unbounded growth in $\tilde{d}(\mathcal{G}_n)$ with increasing size is only a necessary condition and not a sufficient condition for $ \mathcal{H}^*(\mathcal{G}_n)$ to approach zero.

\begin{theorem}
\label{confrag}For any constant $\alpha> 1$, there exist infinite sequences of connected undirected graphs with $n$ nodes, $\{\mathcal{G}_n\}_{n\in\mathbb{N}}$, such that 
\begin{equation}
\label{confrag1}
\lim_{n \to \infty} \frac{\tilde{d}(\mathcal{G}_n)}{n-1} \geq \frac{1}{\alpha},\; \;
\lim_{n \to \infty}\mathcal{H}^*(\mathcal{G}_n) =  \infty.
\end{equation}
\end{theorem}
\begin{proof}
We prove this result by designing such a sequence of graphs. Consider the following bridging operation, $\bigoplus$, which connects two disjoint graphs, $\Gamma=(U,F)$ with the node set ${U=\{u_1, \hdots, u_p\}}$ and $\Gamma'=(U',F')$ with the node set ${U'=\{u_1', \hdots, u_q'\}}$, with a single edge such that $\mathcal{G}=(V,E)=\Gamma\bigoplus\Gamma'$ is defined as 
\begin{equation}
\label{confrag3}
V= U \cup U', \; E= F \cup F' \cup \{(u_i,u_j')\},
\end{equation}
for some $u_i \in U$ and $u_j' \in U'$.
It was shown in \cite{gago2018kirchhoff} that the Kirchhoff index of of such a bridged graph can be expressed in terms of the Kirchhoff indices of the two components as 
\begin{equation}
\label{confrag4}
K_f(\mathcal{G})=\dfrac{p+q}{p}K_f(\Gamma)+ \dfrac{p+q}{q}K_f(\Gamma')+ \dfrac{2p^2-3p+1}{6p}+ \dfrac{q-1}{q^2}+1.
\end{equation}
Furthermore, using \eqref{degav2} it can be shown that 
\begin{equation}
\label{confrag5}
\tilde{d}(\mathcal{G})=\dfrac{p\tilde{d}(\Gamma)+q\tilde{d}(\Gamma')+2}{p+q}.
\end{equation}

Now, consider any sequence of graphs $\mathcal{G}_n=\Gamma_{p}\bigoplus\Gamma'_{q}$ such that $p+q=n$, ${K_f(\Gamma'_{q}) = \mathcal{O}(q^3)}$, $p = \lceil n/ \beta \rceil$, and ${\tilde{d}(\Gamma_{p}) \geq (p-1)/ \beta}$ for some constant $\beta >1$ such that $\beta^3 \leq \alpha$. For instance, $\Gamma_{p}$ can be any graph with $p$ nodes and at least $p(p-1)/2\beta$ edges, and $\Gamma'_{q}$ can be a path or a cycle with $q$ nodes, which results in ${K_f(\Gamma'_{q}) = \mathcal{O}(q^3)}$ as per \eqref{robPCSK} and \eqref{srobKf}. For example, if $\Gamma_{p}$ is a complete graph and $\Gamma'_{q}$ is a path, the resulting graph $\Gamma_p\bigoplus\Gamma'_q$ is known as a lollipop graph. For $p=\lceil n/\beta \rceil$, we have $p \approx n/\beta$ and $q\approx n- n/\beta$, which implies  ${K_f(\Gamma'_{q}) = \mathcal{O}(q^3)=\mathcal{O}(n^3)} $ since $\beta >1$ is a constant. Accordingly, using \eqref{srobKf} and \eqref{confrag4}, one can show that such a sequence of graphs $\{\mathcal{G}_n\}_{n\in\mathbb{N}}$ satisfies
\begin{equation}
\label{confrag5c}
\lim_{n \to \infty} \mathcal{H}^*(\mathcal{G}_n)=\infty.
\end{equation}
Furthermore, \eqref{confrag5} and the fact that $p = \lceil n/ \beta \rceil$, $p+q=n$, and ${\tilde{d}(\Gamma_{p}) \geq (p-1)/ \beta}$ for some $\beta>1$ such that  $\beta^3 \leq \alpha$ imply
\begin{equation}
\label{confrag7}
\lim_{n \to \infty} \frac{\tilde{d}(\mathcal{G}_n)}{n-1}  \geq \lim_{n \to \infty} \dfrac{n/\beta(n/\beta-1)}{\beta (n^2-n)}=\frac{1}{\beta^3} \geq \frac{1}{\alpha}.
\end{equation}
Due to \eqref{confrag5c} and \eqref{confrag7}, any such $\{\mathcal{G}_n\}_{n\in\mathbb{N}}$  satisfies \eqref{confrag1}.
\end{proof}

 \subsection{Structural Robustness of Random Regular Graphs}
In this subsection, we show that random $k$-regular graphs are approximate solutions to the combinatorial problem of designing sparse networks with  optimal structural robustness. Our particular focus on random $k$-regular graphs is motivated by their several desirable properties. For example, such graphs yield uniform/bounded node degrees, which is useful in many applications that demand a balanced communication load among the nodes (e.g., \cite{melamed2004araneola,pandurangan2003building}). It is also known that random $k$-regular graphs ($k\geq 3$) are expander graphs, i.e., sparse yet well-connected structures that can not be easily disconnected by a targeted removal of nodes/edges (e.g., \cite{Hoory06}). Moreover, the algebraic connectivity of such graphs is bounded away from zero (e.g., \cite{Friedman03}), which not only implies fast convergence in consensus networks \cite{Olfati04} but also can be used for providing guarantees on their structural robustness. More specifically, as $n$ goes to infinity, for $k \geq 3$ almost every $k$-regular graph has ${\lambda_2(L)\geq k-2\sqrt{k-1}-\epsilon}$ for  any $\epsilon>0$ (e.g., see \cite{Friedman03}). In light of \eqref{hseig}, this property implies an upper bound on the  structural vulnerability  of those graphs since for any graph 
 \begin{equation}
\label{hseig2}
\frac{1}{2n} \sum_{i=2}^n \frac{1}{\lambda_i (L)} \leq \frac{n-1}{2n\lambda_2 (L)}.
\end{equation}
Accordingly, for any integer $k \geq 3$ and $\epsilon \in (0,k-2\sqrt{k-1})$
\begin{equation}
\lim_{n \to \infty} \Pr \left \{\mathcal{H}^*(\mathcal{G}_{n,k}) \leq \frac{n-1}{2n(k-2\sqrt{k-1}-\epsilon)} \right \}  =1,
\label{rreq}
\end{equation}
where $\mathcal{G}_{n,k}$ is a random $k$-regular graph. Since $n$ and $k$ cannot be both odd (the number of edges is equal to $nk/2$), for odd values of $k$ the limit in \eqref{rreq} is defined along the sequence of even integers $n \in \{k+1,k+3, \hdots\}$.  Furthermore, the probability in \eqref{rreq} tends to one rather fast with increasing $n$ even for moderate values of $k$. Using \cite[Theorem 1.1]{Friedman03}, it can be shown that for any even integer $k \geq 4$
\begin{equation}
\Pr \left \{\mathcal{H}^*(\mathcal{G}_{n,k}) \leq \frac{n-1}{2n(k-2\sqrt{k-1}-\epsilon)} \right \}  \geq 1- \frac{c}{n^\tau},
\label{resb1}
\end{equation}
 for some constant $c>0$ and
\begin{equation}
\tau = \left \lceil \frac{\sqrt{k-1}+1}{2} \right \rceil -1.
\label{resb2}
\end{equation}
As such, the bound in \eqref{rreq} is satisfied by a random $k$-regular graph of size $n$ with probability at least $1-\mathcal{O}(1/n)$ for $k=4$, with probability at least $1-\mathcal{O}(1/n^2)$ for $k=12$, and so on.

By combining \eqref{rreq} with the lower bound in \eqref{boundseq} for ${\tilde{d}(\mathcal{G})=k}$, we can show that for large values of $n$, with high probability, the  structural vulnerability of random $k$-regular graphs ($k\geq 3$) is within a constant factor of the smallest possible value among the graphs with the same size and average degree.  Furthermore, this factor gets arbitrarily close to one as $k$ increases. In other words, for large values of $k$, random $k$-regular graphs have a structural robustness arbitrarily close to the best possible value (with that many edges) with high probability as the network size increases.

\begin{theorem}
\label{rand-regt} For any integer $k\geq 3$ and $\epsilon \in (0,k-2\sqrt{k-1})$
\begin{equation}
\lim_{n \to \infty} \Pr \left \{ \frac{\mathcal{H}^*(\mathcal{G}_{n,k})}{\min\limits_{\mathcal{G}_n : \tilde{d}(\mathcal{G}_n)=k}\mathcal{H}^*(\mathcal{G}_n)}\leq \frac{k}{k-2\sqrt{k-1}-\epsilon}+\epsilon \right \}  =1, \; 
\label{rand-regt1}
\end{equation}
where $\mathcal{G}_{n,k}$ is a random $k$-regular graph.
\end{theorem}
\begin{proof}
Using the lower bound in \eqref{boundseq} and \eqref{robPCSK}, for any undirected graph $\mathcal{G}_n$ with $n$ nodes and average degree $\tilde{d}(\mathcal{G}_n)=k$,
\begin{equation}
\label{randregt2}
\min\limits_{\mathcal{G}_n : \tilde{d}(\mathcal{G}_n)=k}\mathcal{H}^*(\mathcal{G}_n) \geq \frac{(n-1)^2}{2kn^2}
\end{equation}
Using \eqref{randregt2} with \eqref{rreq}, for any random $k$-regular graph with $k \geq 3$ and $\epsilon \in (0,k-2\sqrt{k-1})$, 
\begin{equation} 
\lim\limits_{n \to \infty} \Pr \left \{ \frac{\mathcal{H}^*(\mathcal{G}_{n,k})}{\min\limits_{\mathcal{G}_n : \tilde{d}(\mathcal{G}_n)=k}\mathcal{H}^*(\mathcal{G}_n)}\leq \frac{2kn^2}{(2n^2-2n)(k-2\sqrt{k-1}-\epsilon)}\right \}  =1. 
\label{rand-regt3}
\end{equation}
Note that the upper bound in \eqref{rand-regt3} satisfies
\begin{equation} 
\lim\limits_{n \to \infty}  \frac{2kn^2}{(2n^2-2n)(k-2\sqrt{k-1}-\epsilon)}= \frac{k}{k-2\sqrt{k-1}-\epsilon}.
\label{rand-regt4a}
\end{equation}

Due to the definition of limit, for any $\epsilon'>0$ there exists some $n' \in \mathbb{N}$ such that 
\begin{equation} 
\left| \frac{2kn^2}{(2n^2-2n)(k-2\sqrt{k-1}-\epsilon)} - \frac{k}{k-2\sqrt{k-1}-\epsilon} \right | < \epsilon', \forall n > n'.
\label{rl2}
\end{equation}
Note that for any $k \geq 3$, $\epsilon \in (0,k-2\sqrt{k-1})$, and $n\geq2$,
\begin{equation} 
\frac{2kn^2}{(2n^2-2n)(k-2\sqrt{k-1}-\epsilon)}  \geq \frac{k}{k-2\sqrt{k-1}-\epsilon}.
\label{rl3}
\end{equation}
Using \eqref{rl2} and \eqref{rl3}, we get
\begin{equation} 
\frac{2kn^2}{(2n^2-2n)(k-2\sqrt{k-1}-\epsilon)}  < \frac{k}{k-2\sqrt{k-1}-\epsilon} + \epsilon', \forall n > n'.
\label{rand-regt4}
\end{equation}
Without loss of generality, we can pick $\epsilon'=\epsilon$ and \eqref{rand-regt4} would hold for the corresponding $n' \in \mathbb{N}$. Accordingly, we can obtain \eqref{rand-regt1} from \eqref{rand-regt3} and \eqref{rand-regt4}.

\end{proof}
 In light of Theorem \ref{rand-regt}, random $k$-regular graphs with sufficiently large degree $k$ and size $n$ typically have bounded suboptimality in their structural robustness when compared to the best graph of size $n$ and average degree $k$. Fig. \ref{fig-approx} illustrates how the approximation bound in \eqref{rand-regt1} changes as a function of $k$. As shown in this figure, the approximation bound starts around 17.5 for $k=3$, rapidly drops to 5 by $k=5$ and to 2 by $k=15$, and then keeps approaching one as $k$ increases.  Accordingly, random $k$-regular graphs with $n$ nodes are typically very good approximate solutions to the problem of optimizing structural robustness subject to a sparsity constraint $\tilde{d}(\mathcal{G}_n)=k$ for sufficiently large values of $n$ and $k$  .  We complement Theorem \ref{rand-regt} by showing that the  structural vulnerability  of random $k$-regular graphs is typically within a bounded proximity of the complete graph's structural vulnerability for sufficiently large $k$ and $n$.

\begin{figure}
\begin{center}
\includegraphics[trim = 0mm 0mm 0mm 0mm,clip,scale=0.34]{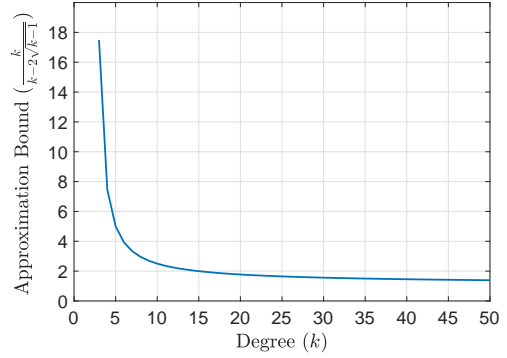}
\caption{Approximation bound in \eqref{rand-regt1}, which bounds the ratio of the  structural vulnerability  of random $k$-regular graphs to the smallest possible value among the graphs with same size and average degree $k$, is shown as a function of $k$. 
}
\label{fig-approx}
\end{center}
\end{figure}

\begin{theorem}   \label{rand-reglem} For any constant $\alpha \geq 1$ and any $\epsilon>0$, 
\begin{equation}
\lim_{n \to \infty} \Pr \left \{\frac{\mathcal{H}^*(\mathcal{G}_{n,k})}{\mathcal{H}^*(\mathcal{K}_n)} \leq\alpha+\epsilon \right \}  =1, \; \forall k \geq \frac{n-1}{\alpha},
\label{rand-reg1}
\end{equation}
where $\mathcal{K}_n$ is the complete graph and $\mathcal{G}_{n,k}$ is a random $k$-regular graph. 
\end{theorem}
\begin{proof}

Note that the denominator of the upper bound in \eqref{rreq} is strictly increasing in $k$ since
\begin{equation}
    \dfrac{d(k-2\sqrt{k-1}-\epsilon)}{dk}=1-\dfrac{1}{\sqrt{k-1}}>0, \; \forall k \geq 3.
\end{equation}
Accordingly, for ${k \geq (n-1)/\alpha}$, we can plug the smallest possible value of $k$ into the upper bound in \eqref{rreq} and obtain 

\begin{equation}
\lim_{n \to \infty} \Pr \left \{\mathcal{H}^*(\mathcal{G}_{n,k}) \leq  \dfrac{n-1} {2n \left ( \dfrac{n-1}{\alpha} -2\sqrt{\dfrac{n-1-\alpha}{\alpha}}-\epsilon\right )} \right \}  =1,
\label{rand-reg2}
\end{equation}
for any $\epsilon \in (0,k-2\sqrt{k-1})$. 
Using \eqref{rand-reg2} together with \eqref{robPCSK}, and without loss of generality setting $\epsilon=0.1$, which is in $(0,k-2\sqrt{k-1})$ for all $k\geq 3$, we have
\begin{equation}
%\resizebox{1\hsize}{!}{
\lim_{n \to \infty} \Pr \left \{\dfrac{\mathcal{H}^*(\mathcal{G}_{n,k})}{\mathcal{H}^*(\mathcal{K}_n)} \leq \dfrac{n} { \dfrac{n-1}{\alpha} -2\sqrt{\dfrac{n-1-\alpha}{\alpha}}-0.1} \right \}  =1.
%}
\label{rand-reg3}
\end{equation}
Note that the upper bound in \eqref{rand-reg3} approaches $\alpha$ as $n \to \infty$. Hence, for any $\epsilon>0$ there is a sufficiently large value of $n$ such that, the upper bound is smaller than $\alpha +\epsilon$. Accordingly, we obtain \eqref{rand-reg1}.
\end{proof}

Theorem \ref{rand-reglem} implies that the random regular graphs can approach the fundamental limit in \eqref{gn-hcomp-cor} on  $\mathcal{H}^*(\mathcal{G}_n)/\mathcal{H}^*(\mathcal{K}_n)$ imposed by the sparsity of $\mathcal{G}_n$. For example, for any constant $\alpha \geq 1$ and even number of nodes $n$ such that $n \geq 3\alpha +1$, let $\mathcal{G}_{n,k^*}$ be a random $k^*$-regular graph, where
\begin{equation}
    \label{kstar}
    k^*=\left \lceil \frac{n-1}{\alpha} \right \rceil .
\end{equation}
 For such random regular graphs, as $n$ increases,  $\mathcal{H}^*(\mathcal{G}_{n,k^*})/\mathcal{H}^*(\mathcal{K}_n)$ is upper bounded by $\alpha$ with a very high probability due to \eqref{rand-reg1}. Furthermore, $\mathcal{G}_{n,k^*}$ has an average degree of $k^*$ that is equal or very close to the minimum required value of $(n-1)/\alpha$ as given in \eqref{gn-hcomp-cor}.

\section{Simulation Results}
\label{sims}

We simulate the noisy consensus dynamics in \eqref{consensus2}, where $\xi(t) \in \mathbb{R}^n$ is white Gaussian noise with zero mean and unit covariance, on different networks with uniform edge weights $w=\bold{1}$ to demonstrate their structural robustness. In each simulation, the network is initialized at $x(0) = \bold{0}$.

In the first set of simulations, we consider the path, star, random $3$-regular, and complete graphs. We generate the random regular graphs using the distributed algorithm in \cite{Yasin15TNSE}.  We aim to numerically illustrate how the structural robustness of these graphs compare to each other and change with increasing network size. For each type we generate three networks of different sizes: $n=20$, $n=40$, and $n=60$. The resulting state variances over time on the networks with $n=60$ are shown in Fig.~\ref{fig-60}. In Table~\ref{Tab-PSRC}, for each of these networks we provide the average of state variance over the simulation horizon and the theoretical value of structural vulnerability, which is computed using the Laplacian eigenvalues as per \eqref{hseig}. For the path, star, and complete graphs, the empirical values can also be verified using \eqref{robPCSK}. For the random $3$-regular graphs, the average distances are computed as $2.62$ ($n=20$), $3.62$ ($n=40$), and $4.09$ ($n=60$). Using the average distances together with the average degrees, the lower and upper bounds in \eqref{boundseq} are computed as $0.15$ and $0.62$ ($\mathcal{G}_{20,3}$), $0.158$ and $0.882$ ($\mathcal{G}_{40,3}$), $0.161$ and $1.005$ ($\mathcal{G}_{60,3}$). For each random 3-regular graph, the observed average state variance is inside the corresponding interval, closer to the lower bound.  
\begin{figure}
\begin{center}
\includegraphics[trim = 0mm 0mm 0mm 0mm,clip,scale=0.43]{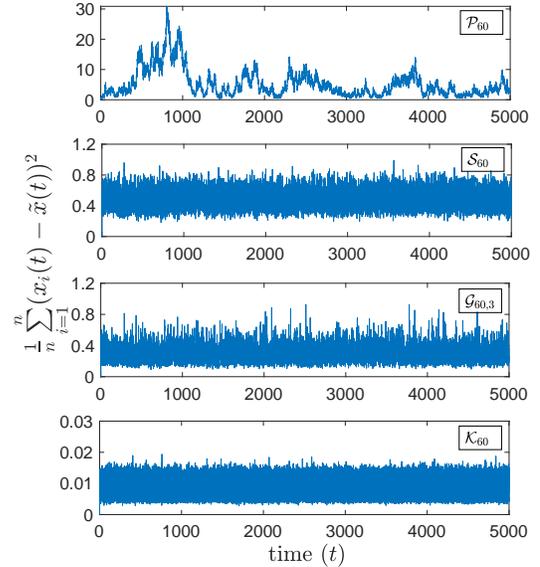}
\vspace*{-5mm} 
\caption{Variance of states under the noisy consensus dynamics  on path ($\mathcal{P}_{60}$),  star ($\mathcal{S}_{60}$), random $3$-regular ($\mathcal{G}_{60,3}$), and complete ($\mathcal{K}_{60}$) graphs with 60 nodes. The edge weights are all set to one to illustrate structural robustness.
}
\label{fig-60}
\end{center}
\end{figure}

\begin{table}\centering
%\resizebox{0.29\textwidth}{!}{
{\renewcommand{\arraystretch}{1}\footnotesize
\begin{center}
 \begin{tabular}{|c |c |c |c|} 
 \cline{2-4}
 %\hline
 \multicolumn{1}{c|}{}
 & $n=20$ & $n=40$ & $n=60$ \\ [0.5ex] 
 \hline
 \multirow{2}{4em}{\centering Path } & 1.662 & 3.32 & 5.21\\ 
& \bf{1.663} & \bf{3.33} & \bf{4.99} \\ 
 \hline
 \multirow{2}{4em}{\centering Star} & 0.45 & 0.478 & 0.485\\ 
& \bf{0.45} & \bf{0.475} & \bf{0.483}\\ 
 \hline 
  \multirow{2}{4em}{\centering Random 3-regular} & 0.239 & 0.286 & 0.305\\ 
& \bf{0.237} & \bf{0.287} & \bf{0.305} \\ 
 \hline
 \multirow{2}{4em}{\centering Complete} & 0.024 & 0.0124 & 0.0085\\ 
& \bf{0.024} & \bf{0.0122} & \bf{0.0082}\\ 
 \hline 
\end{tabular}
\end{center}
}
\caption{Average of state variances over time and the value of  structural vulnerability  as per \eqref{hseig} (bold) for the path, star, random 3-regular, and complete graphs of sizes 20,40, and 60.}
  \label{Tab-PSRC}
\end{table}

In the second set of simulations, we aim to illustrate how the structural robustness of random $k$-regular graphs with $k$ as in \eqref{kstar} change with increasing network size for a given $\alpha \geq 1$. As such, we investigate the performance of such graphs as an approximate solution to the combinatorial problem of designing a network with minimum sparsity that has  $\mathcal{H}^*(\mathcal{G}_n) \leq \alpha \mathcal{H}^*(\mathcal{K}_n)$. For this simulation we pick $\alpha=25$ and set $k$ as per \eqref{kstar} for four different sizes: $n=100$, $n=150$, $n=200$, and $n=250$. Accordingly, we simulate the noisy consensus dynamics on the random regular graphs $\mathcal{G}_{100,4}$, $\mathcal{G}_{150,6}$, $\mathcal{G}_{200,8}$, and $\mathcal{G}_{250,10}$. 
The average of state variances over the simulation horizon were observed as 
$0.1818$ ($\mathcal{G}_{100,4}$), $0.1015$ ($\mathcal{G}_{150,6}$), $0.0717$ ($\mathcal{G}_{200,8}$), and $0.0556$ ($\mathcal{G}_{250,10}$). 
In Table \ref{Tab-RR}, we provide the theoretical values of $\mathcal{H}^*(\mathcal{G}_{n,k})$ and $\mathcal{H}^*(\mathcal{K}_n)$, which are computed using the Laplacian eigenvalues of graphs as per \eqref{hseig}. We also provide their ratios, $\mathcal{H}^*(\mathcal{G}_{n,k})/\mathcal{H}^*(\mathcal{K}_n)$, in the last row of this table. The ratio starts at 36.3 for $n=100$ and monotonically drops to $27.8$ by $n=250$. These results indicate that $\mathcal{H}^*(\mathcal{G}_{n,k})/\mathcal{H}^*(\mathcal{K}_n)$ is approaching $\alpha$ in accordance with Theorem \ref{rand-reglem}. Hence, such random $k$-regular graphs with $k$ as per \eqref{kstar} approximately maintain the required level of robustness with the minimum average degree possible as shown in \eqref{gn-hcomp-cor}.

\begin{table}\centering
%\resizebox{0.38\textwidth}{!}{
{\renewcommand{\arraystretch}{1}\footnotesize
\begin{center}
 \begin{tabular}{|c |c |c |c| c|} 
 \cline{2-5}
 %\hline
\multicolumn{1}{c|}{}
 
 & $n=100$ & $n=150$ & $n=200$ & $n=250$ \\  \multicolumn{1}{c|}{}& $k=4$ & $k=6$ & $k=8$ & $k=10$\\
 \hline
 
  \multirow{2}{4em}{\centering Random $k$-regular} & \multirow{2}{3em}{\centering 0.1813} & \multirow{2}{3em}{\centering 0.1013} & \multirow{2}{3em}{\centering 0.0716}& \multirow{2}{3em}{\centering 0.0555}\\ 
&  &  & &\\

 \hline
 \multirow{2}{4em}{\centering Complete} &  \multirow{2}{3em}{\centering 0.005} & \multirow{2}{3em}{\centering 0.0033} & \multirow{2}{3em}{\centering 0.0025}& \multirow{2}{3em}{\centering 0.002}\\ 
&  &  & &\\ 
 \hline 
  \multirow{2}{4em}{\centering Ratio } &  \multirow{2}{3em}{\centering 36.3} & \multirow{2}{3em}{\centering 30.7} & \multirow{2}{3em}{\centering 28.6}& \multirow{2}{3em}{\centering 27.8}\\ 
&  &  & &\\ 
 \hline 
\end{tabular}
\end{center}
}
\caption{structural vulnerability of the random $k$-regular graphs and the complete graphs of size $n$ as per \eqref{hseig} and their ratios.}
  \label{Tab-RR}
\end{table}

\section{Conclusion}
\label{conclusion}

We investigated the structural robustness of undirected linear consensus networks to noisy interactions. We measured the structural robustness of a graph based on the smallest possible value of the expected steady state population variance of states under the noisy consensus dynamics with admissible edge weights in $(0,1]$. We showed that the average distance and the average node degree in the underlying graph define tight bounds on the structural robustness. Using these novel bounds, we also presented some fundamental graph topological limitations on structural robustness and we investigated the graphs with extremal robustness properties.

As a future direction, we intend to extend our robustness analysis to the generalized case of directed graphs, where the interactions between nodes are not necessarily symmetric. We also plan to investigate the fundamental trade-offs between the proposed measure of structural robustness and other system properties. For example, recently it was shown that the distances between the nodes have a major impact on the controllability of consensus networks and there are trade-offs between the controllability and robustness of such systems (e.g., \cite{Yasin16TAC, Abbas19}). We believe that the results in this paper can be used for further investigation of such relationships between important system properties.

\bibliography{MyReferences}

\end{document}